\documentclass[conference,10pt]{IEEEtran}
\IEEEoverridecommandlockouts
\usepackage{cite}
\usepackage[T1]{fontenc}
\usepackage{graphicx}

\usepackage{amssymb}
\usepackage{amsmath}

\usepackage{paralist}
\usepackage{subfigure}
\usepackage{booktabs} 
\usepackage{algorithm}
\usepackage{algorithmic}
\usepackage{amsthm}
\usepackage{multirow}
\usepackage{microtype} 
\usepackage{tablefootnote}
\usepackage{balance}
\usepackage{framed}
\usepackage{url}
\usepackage{xcolor}
\usepackage[paper=letterpaper,top=0.7in,bottom=1.0in,right=0.68in,left=0.68in]{geometry}

\usepackage{amssymb}
\usepackage{amsfonts}
\usepackage{mathrsfs}
\usepackage{xspace}
\usepackage{bm}
\usepackage{upgreek}

\newcommand{\safemath}[2]{\newcommand{#1}{\ensuremath{#2}\xspace}}

\safemath{\bma}{\mathbf{a}}
\safemath{\bmb}{\mathbf{b}}
\safemath{\bmc}{\mathbf{c}}
\safemath{\bmd}{\mathbf{d}}
\safemath{\bme}{\mathbf{e}}
\safemath{\bmf}{\mathbf{f}}
\safemath{\bmg}{\mathbf{g}}
\safemath{\bmh}{\mathbf{h}}
\safemath{\bmi}{\mathbf{i}}
\safemath{\bmj}{\mathbf{j}}
\safemath{\bmk}{\mathbf{k}}
\safemath{\bml}{\mathbf{l}}
\safemath{\bmm}{\mathbf{m}}
\safemath{\bmn}{\mathbf{n}}
\safemath{\bmo}{\mathbf{o}}
\safemath{\bmp}{\mathbf{p}}
\safemath{\bmq}{\mathbf{q}}
\safemath{\bmr}{\mathbf{r}}
\safemath{\bms}{\mathbf{s}}
\safemath{\bmt}{\mathbf{t}}
\safemath{\bmu}{\mathbf{u}}
\safemath{\bmv}{\mathbf{v}}
\safemath{\bmw}{\mathbf{w}}
\safemath{\bmx}{\mathbf{x}}
\safemath{\bmy}{\mathbf{y}}
\safemath{\bmz}{\mathbf{z}}
\safemath{\bmzero}{\mathbf{0}}
\safemath{\bmone}{\mathbf{1}}

\bmdefine{\biad}{a}
\bmdefine{\bibd}{b}
\bmdefine{\bicd}{c}
\bmdefine{\bidd}{d}
\bmdefine{\bied}{e}
\bmdefine{\bifd}{f}
\bmdefine{\bigd}{g}
\bmdefine{\bihd}{h}
\bmdefine{\biid}{i}
\bmdefine{\bijd}{j}
\bmdefine{\bikd}{k}
\bmdefine{\bild}{l}
\bmdefine{\bimd}{m}
\bmdefine{\bind}{n}
\bmdefine{\biod}{o}
\bmdefine{\bipd}{p}
\bmdefine{\biqd}{q}
\bmdefine{\bird}{r}
\bmdefine{\bisd}{s}
\bmdefine{\bitd}{t}
\bmdefine{\biud}{u}
\bmdefine{\bivd}{v}
\bmdefine{\biwd}{w}
\bmdefine{\bixd}{x}
\bmdefine{\biyd}{y}
\bmdefine{\bizd}{z}

\bmdefine{\bixid}{\xi}
\bmdefine{\bilambdad}{\lambda}
\bmdefine{\bimud}{\mu}
\bmdefine{\bithetad}{\theta}
\bmdefine{\biphid}{\phi}
\bmdefine{\bideltad}{\delta}

\safemath{\bmia}{\biad}
\safemath{\bmib}{\bibd}
\safemath{\bmic}{\bicd}
\safemath{\bmid}{\bidd}
\safemath{\bmie}{\bied}
\safemath{\bmif}{\bifd}
\safemath{\bmig}{\bigd}
\safemath{\bmih}{\bihd}
\safemath{\bmii}{\biid}
\safemath{\bmij}{\bijd}
\safemath{\bmik}{\bikd}
\safemath{\bmil}{\bild}
\safemath{\bmim}{\bimd}
\safemath{\bmin}{\bind}
\safemath{\bmio}{\biod}
\safemath{\bmip}{\bipd}
\safemath{\bmiq}{\biqd}
\safemath{\bmir}{\bird}
\safemath{\bmis}{\bisd}
\safemath{\bmit}{\bitd}
\safemath{\bmiu}{\biud}
\safemath{\bmiv}{\bivd}
\safemath{\bmiw}{\biwd}
\safemath{\bmix}{\bixd}
\safemath{\bmiy}{\biyd}
\safemath{\bmiz}{\bizd}

\safemath{\bmxi}{\bixid}
\safemath{\bmlambda}{\bilambdad}
\safemath{\bmmu}{\bimud}
\safemath{\bmtheta}{\bithetad}
\safemath{\bmphi}{\biphid}
\safemath{\bmdelta}{\bideltad}

\safemath{\bA}{\mathbf{A}}
\safemath{\bB}{\mathbf{B}}
\safemath{\bC}{\mathbf{C}}
\safemath{\bD}{\mathbf{D}}
\safemath{\bE}{\mathbf{E}}
\safemath{\bF}{\mathbf{F}}
\safemath{\bG}{\mathbf{G}}
\safemath{\bH}{\mathbf{H}}
\safemath{\bI}{\mathbf{I}}
\safemath{\bJ}{\mathbf{J}}
\safemath{\bK}{\mathbf{K}}
\safemath{\bL}{\mathbf{L}}
\safemath{\bM}{\mathbf{M}}
\safemath{\bN}{\mathbf{N}}
\safemath{\bO}{\mathbf{O}}
\safemath{\bP}{\mathbf{P}}
\safemath{\bQ}{\mathbf{Q}}
\safemath{\bR}{\mathbf{R}}
\safemath{\bS}{\mathbf{S}}
\safemath{\bT}{\mathbf{T}}
\safemath{\bU}{\mathbf{U}}
\safemath{\bV}{\mathbf{V}}
\safemath{\bW}{\mathbf{W}}
\safemath{\bX}{\mathbf{X}}
\safemath{\bY}{\mathbf{Y}}
\safemath{\bZ}{\mathbf{Z}}

\safemath{\bZero}{\mathbf{0}}
\safemath{\bOne}{\mathbf{1}}
\safemath{\bDelta}{\mathbf{\Delta}}
\safemath{\bLambda}{\mathbf{\UpLambda}}
\safemath{\bPhi}{\mathbf{\Upphi}}
\safemath{\bSigma}{\mathbf{\Upsigma}}
\safemath{\bOmega}{\mathbf{\Upomega}}
\safemath{\bTheta}{\mathbf{\Uptheta}}

\bmdefine{\biAd}{A}
\bmdefine{\biBd}{B}
\bmdefine{\biCd}{C}
\bmdefine{\biDd}{D}
\bmdefine{\biEd}{E}
\bmdefine{\biFd}{F}
\bmdefine{\biGd}{G}
\bmdefine{\biHd}{H}
\bmdefine{\biId}{I}
\bmdefine{\biJd}{J}
\bmdefine{\biKd}{K}
\bmdefine{\biLd}{L}
\bmdefine{\biMd}{M}
\bmdefine{\biOd}{N}
\bmdefine{\biPd}{O}
\bmdefine{\biQd}{P}
\bmdefine{\biRd}{R}
\bmdefine{\biSd}{S}
\bmdefine{\biTd}{T}
\bmdefine{\biUd}{U}
\bmdefine{\biVd}{V}
\bmdefine{\biWd}{W}
\bmdefine{\biXd}{X}
\bmdefine{\biYd}{Y}
\bmdefine{\biZd}{Z}

\bmdefine{\biDelta}{\Delta}
\bmdefine{\biLambda}{\Lambda}
\bmdefine{\biPhi}{\Phi}
\bmdefine{\biSigma}{\Sigma}
\bmdefine{\biOmega}{\Omega}
\bmdefine{\biTheta}{\Theta}

\safemath{\bimA}{\biAd}
\safemath{\bimB}{\biBd}
\safemath{\bimC}{\biCd}
\safemath{\bimD}{\biDd}
\safemath{\bimE}{\biEd}
\safemath{\bimF}{\biFd}
\safemath{\bimG}{\biGd}
\safemath{\bimH}{\biHd}
\safemath{\bimI}{\biId}
\safemath{\bimJ}{\biJd}
\safemath{\bimK}{\biKd}
\safemath{\bimL}{\biLd}
\safemath{\bimM}{\biMd}
\safemath{\bimN}{\biNd}
\safemath{\bimO}{\biOd}
\safemath{\bimP}{\biPd}
\safemath{\bimQ}{\biQd}
\safemath{\bimR}{\biRd}
\safemath{\bimS}{\biSd}
\safemath{\bimT}{\biTd}
\safemath{\bimU}{\biUd}
\safemath{\bimV}{\biVd}
\safemath{\bimW}{\biWd}
\safemath{\bimX}{\biXd}
\safemath{\bimY}{\biYd}
\safemath{\bimZ}{\biZd}

\safemath{\bimDelta}{\biDelta}
\safemath{\bimLambda}{\biLambda}
\safemath{\bimPhi}{\biPhi}
\safemath{\bimSigma}{\biSigma}
\safemath{\bimOmega}{\biOmega}
\safemath{\bimTheta}{\biTheta}

\safemath{\setA}{\mathcal{A}}
\safemath{\setB}{\mathcal{B}}
\safemath{\setC}{\mathcal{C}}
\safemath{\setD}{\mathcal{D}}
\safemath{\setE}{\mathcal{E}}
\safemath{\setF}{\mathcal{F}}
\safemath{\setG}{\mathcal{G}}
\safemath{\setH}{\mathcal{H}}
\safemath{\setI}{\mathcal{I}}
\safemath{\setJ}{\mathcal{J}}
\safemath{\setK}{\mathcal{K}}
\safemath{\setL}{\mathcal{L}}
\safemath{\setM}{\mathcal{M}}
\safemath{\setN}{\mathcal{N}}
\safemath{\setO}{\mathcal{O}}
\safemath{\setP}{\mathcal{P}}
\safemath{\setQ}{\mathcal{Q}}
\safemath{\setR}{\mathcal{R}}
\safemath{\setS}{\mathcal{S}}
\safemath{\setT}{\mathcal{T}}
\safemath{\setU}{\mathcal{U}}
\safemath{\setV}{\mathcal{V}}
\safemath{\setW}{\mathcal{W}}
\safemath{\setX}{\mathcal{X}}
\safemath{\setY}{\mathcal{Y}}
\safemath{\setZ}{\mathcal{Z}}
\safemath{\emptySet}{\varnothing}

\safemath{\colA}{\mathscr{A}}
\safemath{\colB}{\mathscr{B}}
\safemath{\colC}{\mathscr{C}}
\safemath{\colD}{\mathscr{D}}
\safemath{\colE}{\mathscr{E}}
\safemath{\colF}{\mathscr{F}}
\safemath{\colG}{\mathscr{G}}
\safemath{\colH}{\mathscr{H}}
\safemath{\colI}{\mathscr{I}}
\safemath{\colJ}{\mathscr{J}}
\safemath{\colK}{\mathscr{K}}
\safemath{\colL}{\mathscr{L}}
\safemath{\colM}{\mathscr{M}}
\safemath{\colN}{\mathscr{N}}
\safemath{\colO}{\mathscr{O}}
\safemath{\colP}{\mathscr{P}}
\safemath{\colQ}{\mathscr{Q}}
\safemath{\colR}{\mathscr{R}}
\safemath{\colS}{\mathscr{S}}
\safemath{\colT}{\mathscr{T}}
\safemath{\colU}{\mathscr{U}}
\safemath{\colV}{\mathscr{V}}
\safemath{\colW}{\mathscr{W}}
\safemath{\colX}{\mathscr{X}}
\safemath{\colY}{\mathscr{Y}}
\safemath{\colZ}{\mathscr{Z}}

\safemath{\opA}{\mathbb{A}}
\safemath{\opB}{\mathbb{B}}
\safemath{\opC}{\mathbb{C}}
\safemath{\opD}{\mathbb{D}}
\safemath{\opE}{\mathbb{E}}
\safemath{\opF}{\mathbb{F}}
\safemath{\opG}{\mathbb{G}}
\safemath{\opH}{\mathbb{H}}
\safemath{\opI}{\mathbb{I}}
\safemath{\opJ}{\mathbb{J}}
\safemath{\opK}{\mathbb{K}}
\safemath{\opL}{\mathbb{L}}
\safemath{\opM}{\mathbb{M}}
\safemath{\opN}{\mathbb{N}}
\safemath{\opO}{\mathbb{O}}
\safemath{\opP}{\mathbb{P}}
\safemath{\opQ}{\mathbb{Q}}
\safemath{\opR}{\mathbb{R}}
\safemath{\opS}{\mathbb{S}}
\safemath{\opT}{\mathbb{T}}
\safemath{\opU}{\mathbb{U}}
\safemath{\opV}{\mathbb{V}}
\safemath{\opW}{\mathbb{W}}
\safemath{\opX}{\mathbb{X}}
\safemath{\opY}{\mathbb{Y}}
\safemath{\opZ}{\mathbb{Z}}
\safemath{\opZero}{\mathbb{O}}
\safemath{\identityop}{\opI}

\safemath{\veca}{\bma}
\safemath{\vecb}{\bmb}
\safemath{\vecc}{\bmc}
\safemath{\vecd}{\bmd}
\safemath{\vece}{\bme}
\safemath{\vecf}{\bmf}
\safemath{\vecg}{\bmg}
\safemath{\vech}{\bmh}
\safemath{\veci}{\bmi}
\safemath{\vecj}{\bmj}
\safemath{\veck}{\bmk}
\safemath{\vecl}{\bml}
\safemath{\vecm}{\bmm}
\safemath{\vecn}{\bmn}
\safemath{\veco}{\bmo}
\safemath{\vecp}{\bmp}
\safemath{\vecq}{\bmq}
\safemath{\vecr}{\bmr}
\safemath{\vecs}{\bms}
\safemath{\vect}{\bmt}
\safemath{\vecu}{\bmu}
\safemath{\vecv}{\bmv}
\safemath{\vecw}{\bmw}
\safemath{\vecx}{\bmx}
\safemath{\vecy}{\bmy}
\safemath{\vecz}{\bmz}

\safemath{\veczero}{\bmzero}
\safemath{\vecone}{\bmone}
\safemath{\vecxi}{\bmxi}
\safemath{\veclambda}{\bmlambda}
\safemath{\vecmu}{\bmmu}
\safemath{\vectheta}{\bmtheta}
\safemath{\vecphi}{\bmphi}
\safemath{\vecdelta}{\bmdelta}

\safemath{\matA}{\bA}
\safemath{\matB}{\bB}
\safemath{\matC}{\bC}
\safemath{\matD}{\bD}
\safemath{\matE}{\bE}
\safemath{\matF}{\bF}
\safemath{\matG}{\bG}
\safemath{\matH}{\bH}
\safemath{\matI}{\bI}
\safemath{\matJ}{\bJ}
\safemath{\matK}{\bK}
\safemath{\matL}{\bL}
\safemath{\matM}{\bM}
\safemath{\matN}{\bN}
\safemath{\matO}{\bO}
\safemath{\matP}{\bP}
\safemath{\matQ}{\bQ}
\safemath{\matR}{\bR}
\safemath{\matS}{\bS}
\safemath{\matT}{\bT}
\safemath{\matU}{\bU}
\safemath{\matV}{\bV}
\safemath{\matW}{\bW}
\safemath{\matX}{\bX}
\safemath{\matY}{\bY}
\safemath{\matZ}{\bZ}
\safemath{\matzero}{\bmzero}

\safemath{\matDelta}{\bDelta}
\safemath{\matLambda}{\bLambda}
\safemath{\matPhi}{\bPhi}
\safemath{\matSigma}{\bSigma}
\safemath{\matOmega}{\bOmega}
\safemath{\matTheta}{\bTheta}

\safemath{\matidentity}{\matI}
\safemath{\matone}{\matO}

\safemath{\rnda}{A}
\safemath{\rndb}{B}
\safemath{\rndc}{C}
\safemath{\rndd}{D}
\safemath{\rnde}{E}
\safemath{\rndf}{F}
\safemath{\rndg}{G}
\safemath{\rndh}{H}
\safemath{\rndi}{I}
\safemath{\rndj}{J}
\safemath{\rndk}{K}
\safemath{\rndl}{L}
\safemath{\rndm}{M}
\safemath{\rndn}{N}
\safemath{\rndo}{O}
\safemath{\rndp}{P}
\safemath{\rndq}{Q}
\safemath{\rndr}{R}
\safemath{\rnds}{S}
\safemath{\rndt}{T}
\safemath{\rndu}{U}
\safemath{\rndv}{V}
\safemath{\rndw}{W}
\safemath{\rndx}{X}
\safemath{\rndy}{Y}
\safemath{\rndz}{Z}

\safemath{\rveca}{\bimA}
\safemath{\rvecb}{\bimB}
\safemath{\rvecc}{\bimC}
\safemath{\rvecd}{\bimD}
\safemath{\rvece}{\bimE}
\safemath{\rvecf}{\bimF}
\safemath{\rvecg}{\bimG}
\safemath{\rvech}{\bimH}
\safemath{\rveci}{\bimI}
\safemath{\rvecj}{\bimJ}
\safemath{\rveck}{\bimK}
\safemath{\rvecl}{\bimL}
\safemath{\rvecm}{\bimM}
\safemath{\rvecn}{\bimN}
\safemath{\rveco}{\bomO}
\safemath{\rvecp}{\bimP}
\safemath{\rvecq}{\bimQ}
\safemath{\rvecr}{\bimR}
\safemath{\rvecs}{\bimS}
\safemath{\rvect}{\bimT}
\safemath{\rvecu}{\bimU}
\safemath{\rvecv}{\bimV}
\safemath{\rvecw}{\bimW}
\safemath{\rvecx}{\bimX}
\safemath{\rvecy}{\bimY}
\safemath{\rvecz}{\bimZ}

\safemath{\rvecxi}{\bmxi}
\safemath{\rveclambda}{\bmlambda}
\safemath{\rvecmu}{\bmmu}
\safemath{\rvectheta}{\bmtheta}
\safemath{\rvecphi}{\bmphi}

\safemath{\rmatA}{\bimA}
\safemath{\rmatB}{\bimB}
\safemath{\rmatC}{\bimC}
\safemath{\rmatD}{\bimD}
\safemath{\rmatE}{\bimE}
\safemath{\rmatF}{\bimF}
\safemath{\rmatG}{\bimG}
\safemath{\rmatH}{\bimH}
\safemath{\rmatI}{\bimI}
\safemath{\rmatJ}{\bimJ}
\safemath{\rmatK}{\bimK}
\safemath{\rmatL}{\bimL}
\safemath{\rmatM}{\bimM}
\safemath{\rmatN}{\bimN}
\safemath{\rmatO}{\bimO}
\safemath{\rmatP}{\bimP}
\safemath{\rmatQ}{\bimQ}
\safemath{\rmatR}{\bimR}
\safemath{\rmatS}{\bimS}
\safemath{\rmatT}{\bimT}
\safemath{\rmatU}{\bimU}
\safemath{\rmatV}{\bimV}
\safemath{\rmatW}{\bimW}
\safemath{\rmatX}{\bimX}
\safemath{\rmatY}{\bimY}
\safemath{\rmatZ}{\bimZ}

\safemath{\rmatDelta}{\bimDelta}
\safemath{\rmatLambda}{\bimLambda}
\safemath{\rmatPhi}{\bimPhi}
\safemath{\rmatSigma}{\bimSigma}
\safemath{\rmatOmega}{\bimOmega}
\safemath{\rmatTheta}{\bimTheta}

\usepackage{amssymb}
\usepackage{amsfonts}
\usepackage{mathrsfs}
\usepackage{xspace}
\usepackage{bm}
\usepackage{fancyref}
\usepackage{textcomp}

\usepackage{multirow}
\usepackage{stmaryrd}

\newenvironment{textbmatrix}{	\setlength{\arraycolsep}{2.5pt}%
								\big[\begin{matrix}}{\end{matrix}\big]%
								\raisebox{0.08ex}{\vphantom{M}}}

\def\be{\begin{equation}}
\def\ee{\end{equation}}
\def\een{\nonumber \end{equation}}
\def\mat{\begin{bmatrix}}
\def\emat{\end{bmatrix}}
\def\btm{\begin{textbmatrix}}
\def\etm{\end{textbmatrix}}

\def\ba#1\ea{\begin{align}#1\end{align}}
\def\bas#1\eas{\begin{align*}#1\end{align*}}
\def\bs#1\es{\begin{split}#1\end{split}}
\def\bg#1\eg{\begin{gather}#1\end{gather}}
\def\bml#1\eml{\begin{multline}#1\end{multline}}
\def\bi#1\ei{\begin{itemize}#1\end{itemize}}

\newcommand{\lefto}{\mathopen{}\left}

\DeclareMathOperator{\Exop}{\opE}			% expectation operator
\newcommand{\Ex}[2]{\ensuremath{\Exop_{#1}\lefto[#2\right]}} 	% expectation
\safemath{\dirac}{\delta}					% Dirac delta
\safemath{\krond}{\dirac}					% Kronecker delta
\safemath{\upto}{\uparrow}
\safemath{\downto}{\downarrow}
\safemath{\iu}{j}							% imaginary unit
\safemath{\ev}{\lambda}						% eigenvalue
\safemath{\hilseqspace}{l^{2}}				% Hilbert sequence space
\newcommand{\banachfunspace}[1]{\setL^{#1}}	% Banach function space
\safemath{\hilfunspace}{\banachfunspace{2}}	% Hilbert function space
\safemath{\SNR}{\textit{SNR}} 				% signal to noise ratio
\safemath{\PAR}{\textit{PAR}} 				% signal to noise ratio
\safemath{\No}{N_0}							% noise spectral density
\safemath{\Es}{E_s}							% energy per symbol
\safemath{\Eb}{E_b}							% energy per bit
\safemath{\EbNo}{\frac{\Eb}{\No}}
\safemath{\EsNo}{\frac{\Es}{\No}}

\DeclareMathOperator{\CHop}{\ensuremath{\opH}} % channel operator
\safemath{\tvir}{\rndh_{\CHop}}				% time-varying impulse response
\safemath{\tvtf}{\rndl_{\CHop}}				% 	-''- transfer function
\safemath{\spf}{\rnds_{\CHop}}				% spreading function
\safemath{\bff}{H_{\CHop}}					% bi-freuqency function

\safemath{\ircf}{r_{h}}						% impulse response correlation fn.
\safemath{\tftvcf}{r_{s}}					% scattering function
\safemath{\tfcf}{r_{l}}						% time-frequency correlation fn.
\safemath{\bfcf}{r_{H}}						% bi-frequency correlation fn.

\safemath{\tcorr}{c_h}						% time-correlation function
\safemath{\scf}{c_{s}}						% spreading function
\safemath{\tfcorr}{c_{l}}					% transfer-function correlation
\safemath{\fcorr}{c_{H}}						% frequency-correlation function

\safemath{\mi}{I}							% mutual information
\safemath{\capacity}{C}						% capacity

\safemath{\normal}{\mathcal{N}}			% normal distribution
\safemath{\jpg}{\mathcal{CN}}			% jointly proper Gaussian
\safemath{\mchain}{\leftrightarrow}		% Markov chain
\safemath{\dB}{\,\mathrm{dB}}
\safemath{\dBm}{\,\mathrm{dBm}}
\safemath{\Hz}{\,\mathrm{Hz}}
\safemath{\kHz}{\,\mathrm{kHz}}
\safemath{\MHz}{\,\mathrm{MHz}}
\safemath{\GHz}{\,\mathrm{GHz}}
\safemath{\s}{\,\mathrm{s}}
\safemath{\ms}{\,\mathrm{ms}}
\safemath{\mus}{\,\mathrm{\text{\textmu}s}}
\safemath{\ns}{\,\mathrm{ns}}
\safemath{\ps}{\,\mathrm{ps}}
\safemath{\meter}{\,\mathrm{m}}
\safemath{\mm}{\,\mathrm{mm}}
\safemath{\cm}{\,\mathrm{cm}}
\safemath{\m}{\,\mathrm{m}}
\safemath{\W}{\,\mathrm{W}}
\safemath{\mW}{\, \mathrm{mW}}
\safemath{\J}{\,\mathrm{J}}
\safemath{\K}{\,\mathrm{K}}
\safemath{\bit}{\,\mathrm{bit}}
\safemath{\nat}{\,\mathrm{nat}}

\safemath{\define}{\triangleq}			% definition

\safemath{\equivalent}{\sim}
\safemath{\distas}{\sim}					% distributed according to
\safemath{\sdiff}{\Delta}				% symmetric set difference

\safemath{\reals}{\mathbb{R}}
\safemath{\positivereals}{\reals_{+}}
\safemath{\integers}{\mathbb{Z}}
\safemath{\posint}{\integers_{+}}
\safemath{\naturals}{\mathbb{N}}
\safemath{\posnaturals}{\naturals_{+}}
\safemath{\complexset}{\mathbb{C}}
\safemath{\rationals}{\mathbb{Q}}

\newcommand*{\fancyrefapplabelprefix}{app}		% Appendix
\newcommand*{\fancyrefthmlabelprefix}{thm}		% Theorem
\newcommand*{\fancyreflemlabelprefix}{lem}		% Lemma
\newcommand*{\fancyrefcorlabelprefix}{cor}		% Corollary
\newcommand*{\fancyrefdeflabelprefix}{def}		% Definition
\newcommand*{\fancyrefproplabelprefix}{prop}		% Proposition
\newcommand*{\fancyrefexmpllabelprefix}{exmpl}
\newcommand*{\fancyrefalglabelprefix}{alg}		% Algorithm
\newcommand*{\fancyreftbllabelprefix}{tbl}		% Algorithm
\newcommand*{\fancyreftestlabelprefix}{est}		% Estimator
\newcommand*{\fancyrefsyslabelprefix}{sys}		% System Model

\frefformat{vario}{\fancyrefseclabelprefix}{Section~#1}
\frefformat{vario}{\fancyrefthmlabelprefix}{Theorem~#1}
\frefformat{vario}{\fancyreftbllabelprefix}{Table~#1}
\frefformat{vario}{\fancyreflemlabelprefix}{Lemma~#1}
\frefformat{vario}{\fancyrefcorlabelprefix}{Corollary~#1}
\frefformat{vario}{\fancyrefdeflabelprefix}{Definition~#1}
\frefformat{vario}{\fancyreffiglabelprefix}{Figure~#1}
\frefformat{vario}{\fancyrefapplabelprefix}{Appendix~#1}
\frefformat{vario}{\fancyrefeqlabelprefix}{(#1)}
\frefformat{vario}{\fancyrefproplabelprefix}{Proposition~#1}
\frefformat{vario}{\fancyrefexmpllabelprefix}{Example~#1}
\frefformat{vario}{\fancyrefalglabelprefix}{Algorithm~#1}
\frefformat{vario}{\fancyreftestlabelprefix}{Estimator~#1}
\frefformat{vario}{\fancyrefsyslabelprefix}{System Model~#1}

\allowdisplaybreaks 
\newcommand{\conditionaltextstyle}{\textstyle}

\newtheorem{theorem}{Theorem}
\newtheorem{lemma}{Lemma}

\newtheorem{definition}{Definition}
\newtheorem{estimator}{Estimator}
\newtheorem{system}{System Model}

\newcommand{\B}[0]{D=128} % symbol for the number of BS antennas
\newcommand{\U}[0]{8} % symbol for the number of UEs

\newcommand{\median}[0]{\mathsf{m}} 
\newcommand{\samplemedian}[0]{\overline{\median}}
\newcommand{\sampleNo}[0]{\overline{N}_0}
\newcommand{\estimatedNo}[0]{\widehat{N}_0}
\newcommand{\sampleEs}[0]{\overline{E}_s}
\newcommand{\estimatedEs}[0]{\widehat{E}_s}
\newcommand{\sampleSNR}[0]{\overline{\SNR}}
\newcommand{\estimatedSNR}[0]{\widehat{\SNR}}

\newcommand{\Eo}[0]{E_0}
\newcommand{\sampleEo}[0]{\overline{E}_0}
\newcommand{\estimatedEo}[0]{\widehat{E}_0}
\newcommand{\Eh}[0]{E_s/p}
\newcommand{\MSE}[0]{\Eo}
\newcommand{\estimatedMSE}[0]{\estimatedEo}

\newcommand{\festnoargs}[0]{\mu}
\newcommand{\fest}[1]{\festnoargs(#1)} 
\newcommand{\abssquared}[1]{|#1|^{2}}

\safemath{\Tran}{\textnormal{T}}
\safemath{\Herm}{\textnormal{H}}

\safemath{\CN}{\mathcal{CN}}
\safemath{\N}{\mathcal{N}}

\safemath{\diag}{\textnormal{diag}}
\safemath{\trace}{\textnormal{trace}}

\begin{document}

\title{Blind SNR Estimation and Nonparametric Channel Denoising in Multi-Antenna mmWave Systems}

\author{\IEEEauthorblockN{Alexandra Gallyas-Sanhueza$^\text{1}$ and Christoph Studer$^\text{2}$} \\[-0.0cm]
\IEEEauthorblockA{\textit{$^\text{1}$Department of Electrical and Computer Engineering, Cornell University, Ithaca, NY; email: ag753@cornell.edu} \\
\textit{$^\text{2}$Dept.\ of Information Technology and Electrical Engineering, ETH Zurich, Zurich, Switzerland; email: studer@ethz.ch}\\[-0.32cm]
} 
\thanks{The work of AGS and CS was supported by ComSenTer, one of six centers in JUMP, a Semiconductor Research Corporation (SRC) program sponsored by DARPA. The work of CS was also supported by an ETH Research Grant and by the US NSF under grants CNS-1717559 and ECCS-1824379.}}

\maketitle

%%%%%%%%%%%%%%%%%%%%%%%%%%%%%%%%%%%%%%%%%%%%%%%%%%%%%%%%%%%%%%%%%%%
%%%%%%%%%%%%%%%%%%%%%%%%%%%%%%%%%%%%%%%%%%%%%%%%%%%%%%%%%%%%%%%%%%%

\begin{abstract}
We propose blind estimators for the average noise power, receive signal power, signal-to-noise ratio (SNR), and mean-square error (MSE), suitable for multi-antenna millimeter wave (mmWave) wireless systems. 
The proposed estimators can be computed at low complexity and solely rely on beamspace sparsity, i.e., the fact that only a small number of dominant propagation paths exist in typical mmWave channels. Our estimators can be used (i) to quickly track some of the key quantities in multi-antenna mmWave systems while avoiding additional pilot overhead and (ii) to design efficient nonparametric algorithms that require such quantities. 
We provide a theoretical analysis of the proposed estimators, and we demonstrate their efficacy via synthetic experiments and using a nonparametric channel-vector denoising task with realistic multi-antenna mmWave channels. 
\end{abstract}

%%%%%%%%%%%%%%%%%%%%%%%%%%%%%%%%%%%%%%%%%%%%%%%%%%%%%%%%%%%%%%%%%%%
%%%%%%%%%%%%%%%%%%%%%%%%%%%%%%%%%%%%%%%%%%%%%%%%%%%%%%%%%%%%%%%%%%%

% !TEX root = 20ICC_nonparametric.tex
% DO NOT REMOVE THE ABOVE COMMENT!

\section{Introduction}
Accurate knowledge of system parameters, such as average noise power, signal power, or the signal-to-noise ratio (SNR) is critical in communication systems, as many baseband processing tasks use these quantities~\cite{schenk2008rf}.
Most conventional communication systems dedicate a training stage to estimate such system parameters. However, since the propagation conditions can change at fast rates, especially at millimeter-wave (mmWave) frequencies where blockers or interferers may appear quickly~\cite{rappaport13a}, it is important to develop low-complexity solutions that accurately track such parameters.

Fortunately, modern wireless communication systems deal with high-dimensional data. For example, massive multiple-input multiple-output (MIMO) basestations are expected to be equipped with hundreds of antennas, or orthogonal frequency-division multiplexing (OFDM) systems have thousands of subcarriers. 
Since many of the signals in such systems are structured (e.g., exhibit sparsity or are taken from a discrete set), one can design statistical methods that blindly estimate parameters, without the need of a dedicated training stage. 

\subsection{Prior Art in Blind and Nonparametric Estimation}
Blind estimators rely on the signal statistics rather than pilot sequences.
Many of the existing blind noise variance and SNR estimators exploit modulation-specific structure, such as the cyclic prefix redundancy in OFDM~\cite{socheleau09a} or periodicity of synchronization sequences \cite{zivkovic09a}.
Other methods use expectation maximization (EM) with sophisticated statistical models. 
For example, EM has been successfully used for blind SNR estimation~\cite{Das12a} or to recover sparse signals~\cite{huang2020a}.
However, the iterative nature of EM and its relatively high per-iteration complexity renders such methods unsuitable for real-time estimation in multi-antenna mmWave wireless systems that operate with high-dimensional data at gigabit per second sampling rates.
In contrast, we propose a range of low-complexity blind estimators whose complexity is only $\setO(D)$, where $D$ is the dimension of the processed data.

In order to deal with algorithm parameters, nonparametric methods have been proposed recently.
The nonparametric equalizer (NOPE)~\cite{ghods2017optimally}, for example, performs linear minimum mean-square error (MSE) estimation  in massive MIMO systems without knowledge of the SNR. NOPE combines approximate message passing~\cite{donoho2009message,rangan11a} with Stein's unbiased risk estimate (SURE) \cite{stein81a} to automatically tune the algorithm parameters.  
The concept of estimating a subset of the algorithm parameters directly from 
the noisy measurements, has been used recently for adaptive denoising of mmWave~\cite{ghods19a,gallyas20a,mirfarshbafan2019beamspace} or OFDM~\cite{upadhya2016risk} channel vectors.
Denoising algorithms typically require two parameters (the noise power and a thresholding parameter), but since the threshold can be estimated using SURE from the noisy observations, such methods only need knowledge of the noise power.
In this paper, we propose low-complexity blind estimators, which enable the design of nonparametric (i.e., parameter free) channel-vector denoisers.

\subsection{Contributions}
Our main contributions are as follows. 
We propose low-complexity blind estimators for the average noise power, signal power, SNR, and MSE after applying an element-wise function to the noisy observation. 
We provide a theoretical analysis by developing bounds on the accuracy in the large-dimension limit that depend on the SNR and sparsity. 
We provide simulation results with synthetic data to demonstrate the efficacy and limits of our estimators in finite dimensions.
We showcase an application example of our blind estimators, which leads to a novel nonparametric channel-vector denoising algorithm for mmWave massive MIMO  systems. 

\subsection{Notation}
Lowercase and uppercase boldface letters denote column vectors and matrices, respectively.
The $d$th entry of the vector~$\bma\in\complexset^D$ is~$a_d$, the real and imaginary parts are $\Re\{\bma\}$ and $\Im\{\bma\}$, respectively, and we use $\bmb=\abssquared{\bma}$ to refer to $b_d=|a_d|^2$ for $d=1,\ldots,D$.
The soft-thresholding function with threshold $\tau$ is defined as $\eta(x;\tau)={x}/{|x|}\max\{|x|-\tau,0\}$ for $x\neq0$ and $\eta(x;\tau)=0$ for $x=0$, and is applied element-wise to vectors.
Statistical quantities are denoted by plain symbols, e.g., the variance $E_x=\frac{1}{D}\Ex{}{\|\bmx\|^2}$ of the random vector $\bmx\in\complexset^D$, where $\Ex{}{\cdot}$ denotes expectation; 
sample estimates are denoted by a bar, e.g., the sample variance $\overline{E}_x = \frac{1}{D}\|\bmx\|^2$; 
corresponding blind estimators are denoted by a hat, e.g.,~$\widehat{E}_x$.
For $x\in\reals$, rounding towards plus and minus infinity is denoted by $\lceil x\rceil$ and $\lfloor x \rfloor$, respectively, and $[x]_+=\max\{x,0\}$.
The big-O notation is $\setO(\cdot)$. 

%
% !TEX root = 20ICC_nonparametric.tex
% DO NOT REMOVE THE ABOVE COMMENT!

\section{Low-Complexity Blind Estimators}
\label{sec:nonparametricestimators}

\subsection{System Models}
\label{sec:systemmodel}
We say that a vector $\bms\in\complexset^D$ is \emph{sparse} if the number of nonzero entries $K$ is smaller than the dimension of the vector~$D$. As a sparsity measure, one can use the $\ell_0$ pseudo norm $\|\bms\|_0=K$, which counts the number of nonzero entries of $\bms$. 
As we will show in \fref{sec:theory}, sparsity is one of the key ingredients to our blind estimators.
In the rest of the paper, we focus on the following two system models.

\begin{system} \label{sys:systemmodel1}
Let $\bms\in\complexset^D$ be a sparse signal with average energy $\Es=\frac{1}{D}\Ex{}{\|\bms\|^2}$.
We model a noisy observation of the sparse signal using the following input-output relation:
\begin{align} \label{eq:inputoutputrelation}
\bmy = \bms + \bmn.
\end{align}
Here, $\bmy\in\complexset^D$ is the noisy observation and $\bmn\in\complexset^D$ models noise with circularly-symmetric i.i.d. complex Gaussian entries with variance $\No$.
We assume that the sparse signal vector $\bms$ and noise vector $\bmn$ are statistically independent. 
\end{system}
In what follows, we do not make assumptions on the signal sparsity $\|\bms\|_0=K$, even though the performance of the proposed estimators depends on this parameter; see \fref{sec:theory} for the details. 
Furthermore, unlike pilot-based estimators, we assume that the sparse vector $\bms$ is unknown, which makes parameter estimation nontrivial in our scenario.

\fref{sys:systemmodel1} finds numerous applications in wireless communication systems. Prime examples are modeling the estimated channel vectors in multi-antenna mmWave systems~\cite{ghods19a,gallyas20a,mirfarshbafan2019beamspace} or in OFDM systems~\cite{upadhya2016risk}, where the beamspace representation or the delay-domain representation of these channel vectors are typically sparse, respectively. 

\begin{system} \label{sys:systemmodel2}
Let $\bmy$ be a noisy observation as in~\fref{sys:systemmodel1}. 
Fix a weakly differentiable function $\festnoargs: \complexset \to \complexset$ that operates element-wise on vectors. We model the output as
\begin{align} \label{eq:inputoutputrelation2}
\fest{\bmy} = \bms + \bme,
\end{align}
where $\bme\in\mathbb{C}^D$ contains residual distortion.
\end{system}
We emphasize that in \fref{eq:inputoutputrelation2}, the residual distortion $\bme$ is not necessarily independent of the sparse signal $\bms$. 
\fref{sys:systemmodel2} is relevant in the following scenarios: 
(i) When estimating a sparse signal $\bms$ from a noisy measurement $\bmy$ by applying an (entry-wise) denoising or estimation function, producing the signal estimate $\tilde\bms=\fest{\bmy}$.
This scenario finds use for channel-vector denosing~\cite{ghods19a, gallyas20a}, for example. 
(ii) When modeling nonlinearlties caused by hardware impairments~\cite{jacobsson18d}, in which case the distorted version of the noisy received signal can be expressed as $\bmr=\fest{\bmy}$.
This scenario finds use in systems with low-resolution data converters~\cite{li17b,jacobsson17b}, for example.

\subsection{Low-Complexity Blind Estimators}
\label{sec:estimatorsusbection}
The blind estimators proposed next make frequent use of the sample median, which we define as follows. 

\begin{definition}[Sample Median]
Let $\bms\in\reals^D$ be a vector and $\bms^\text{sort}\in\reals^D$ be its sorted version (entries sorted in ascending order). Then, the \emph{sample median} is defined as 
\begin{align} \label{eq:samplemedian}
\samplemedian(\bms) = \conditionaltextstyle \frac{1}{2}\Big(s^\text{sort}_{\lfloor(D+1)/2\rfloor}+s^\text{sort}_{\lceil(D+1)/2\rceil}\Big).
\end{align}
\end{definition}
The sample median is {robust to outliers \cite{rousseeuw1993alternatives,huber2004robust},} which makes it amenable to \fref{sys:systemmodel1}, as the nonzero entries of the sparse vector $\bms$ can be considered to be outliers for the purpose of estimating the noise level.
We emphasize that the sample median can be computed at low complexity
in $\setO(D)$ average time using quickselect~\cite{tibshirani09a} or in $\setO(D)$ deterministic time using the MedianOfNinthers algorithm~\cite{alexandrescu17a}.

\begin{estimator}[Average Noise Power] \label{est:noisevariance}
Consider~\fref{sys:systemmodel1}. 
We propose the following blind estimator 
\begin{align} \label{eq:noiseestimator}
\estimatedNo = \conditionaltextstyle  \frac{\samplemedian(\abssquared{\bmy})}{\log(2)},
\end{align}
to estimate the average noise power defined as
\begin{align}
\No= \conditionaltextstyle \frac{1}{D}\Ex{}{\|\bmn\|^2}\!.
\end{align}
\end{estimator}

\fref{est:noisevariance} only requires  the absolute square entries of the noisy observation $\bmy$ in \fref{eq:inputoutputrelation} and can be computed efficiently. 
The estimator exploits sparsity in the signal $\bms$, but is independent of the sparsity level $K$, the signal power, or the statistical sparsity model.
It is, however, important to understand that the accuracy of this estimator depends on all of these factors as it relies on the fact that the nonzero entries of the sparse vector~$\bms$ can be treated as outliers when estimating the average noise power.
We note that this noise variance estimator can be seen as a complex-valued  (and squared) generalization\footnote{The squared median absolute deviation (MAD) \cite{huber2004robust} corresponds to $\samplemedian(|\bmy|)^2$ whereas we propose to use $\samplemedian(|\bmy|^2)$. While for even dimensions $D$ we have $\samplemedian(|\bmy|)^2 \leq \samplemedian(|\bmy|^2)$, both estimators coincide when $D$ is odd. What is more, our scaling factor $\log(2)$ differs from the widely-used scale factor of $(\Phi^{-1}(3/4))^2\approx (0.6745)^2$~\cite{donoho94}. The latter is derived for variance estimation of \emph{real-valued} Gaussians using MAD, while in our derivation we consider the case of \emph{complex-valued} Gaussians.} of the conventional median absolute deviation (MAD)~\cite{rousseeuw1993alternatives}, where we use the assumption that the noise in~\fref{sys:systemmodel1} is zero-mean. 
The intuition behind this estimator (and the $\log(2)$ scaling factor) is the fact that the entries ${|n_d|^2}/{(\No/2)}$, $d=1,\ldots,D$ are  $\chi^2$ distributed with two degrees of freedom, which have a median of $2\log(2)$.
A detailed derivation of this blind estimator and an analysis are provided in \fref{sec:noise_estimator_analysis}.

\begin{estimator}[Average Signal Power] \label{est:signalpower}
Consider~\fref{sys:systemmodel1}. 
We propose the following blind estimator
\begin{align} \label{eq:signalpowerestimator}
\estimatedEs =\conditionaltextstyle  \left[ \frac{\|\bmy\|^2}{D} - \estimatedNo \right]_{\!+}
\end{align}
for the average signal power defined as
 \begin{align}
\Es = \conditionaltextstyle \frac{1}{D}\Ex{}{\|\bms\|^2}\!.
\end{align}
\end{estimator}

\fref{est:signalpower} only uses the sample estimate of the average receive power $\overline{E}_y = \frac{1}{D}\|\bmy\|^2$ and the blind noise estimate~$\estimatedNo$ from \fref{est:noisevariance}. 
The derivation of this estimator and an analysis of its key properties are provided in \fref{sec:signal_estimator_analysis}.

\begin{estimator}[Signal-to-Noise Ratio]  \label{est:SNR}
Consider~\fref{sys:systemmodel1}. 
We propose the following blind estimator
\begin{align} \label{eq:SNRestimateomg}
\estimatedSNR =  \conditionaltextstyle \left[ \frac{\|\bmy\|^2}{D \estimatedNo} - 1 \right]_{\!+}
\end{align}
for the SNR defined as 
\begin{align}
\SNR = \conditionaltextstyle \frac{\Ex{}{\|\bms\|^2}}{\Ex{}{\|\bmn\|^2}}. 
\end{align}
\end{estimator}

\fref{est:SNR} is blind as it combines the sample estimate of the average receive power $\overline{E}_y = \frac{1}{D}\|\bmy\|^2$ and the blind estimate~$\estimatedNo$ from \fref{est:noisevariance}. 
The derivation of this estimator and an analysis of its key properties are provided in \fref{sec:SNR_estimator_analysis}.

\begin{estimator}[Mean-Square Error]   \label{est:MSE}
Consider~\fref{sys:systemmodel2} with a fixed function $\festnoargs: \complexset \to \complexset$.
We propose the following blind estimator  
\begin{align} \label{eq:MSEnonparametricexplicitform}
\estimatedMSE =\, & \conditionaltextstyle \frac{1}{D}\|\fest{\bmy}-\bmy\|_2^2 - \estimatedNo \notag \\
& \conditionaltextstyle + \frac{\estimatedNo}{D}\sum_{d=1}^D\left(\frac{\partial\Re\{\fest{y_d}\}}{\partial\Re\{y_d\}}+\frac{\partial\Im\{\fest{y_d}\}}{\partial\Im\{y_d\}}\right)
\end{align}
for the MSE defined as
\begin{align} \label{eq:MSEexpression}
\MSE = \conditionaltextstyle \frac{1}{D}\Ex{}{\|\fest{\bmy}-\bms\|^2} = \frac{1}{D}\Ex{}{\|\bme\|^2}\!.
\end{align}
\end{estimator}

\fref{est:MSE} only uses the receive signal $\bmy$,  the estimate~$\estimatedNo$ from \fref{est:noisevariance}, and the function $\festnoargs$.
Note that if we use the identity function $\fest{\bmy}=\bmy$, then the MSE corresponds to $\Eo=\No$ while the estimated MSE corresponds to  $\estimatedEo=\estimatedNo$, as expected.
The derivation of this estimator and an analysis of its key properties are provided in \fref{sec:MSE_estimator_analysis}.

% !TEX root = 20ICC_nonparametric.tex
% DO NOT REMOVE THE ABOVE COMMENT!

\section{Theory}
\label{sec:theory}

\subsection{Convergence of the Sample Median for $D\to\infty$}
We will use the following definition of the median.
\begin{definition}[Median]
Let $X$ be an absolutely continuous random variable (RV) with cumulative distribution function (CDF) $F_X(x)$.
Then, the \emph{median} $\median_X$ of $X$ is defined as
\begin{align}
F_X(\median_X) = \conditionaltextstyle  \frac{1}{2}. \label{eq:median}
\end{align}
\end{definition}

We will frequently make use of the following result.

\begin{lemma}[Lemma~C.1 from \cite{MM08}] \label{lem:convergenceofmedian}
Suppose that $f_X(x)$ is a differentiable probability density function (PDF) in some neighborhood of the median $\median_X$, and vector $\bmx$ contains i.i.d.\ samples of $X$.
Then, for any $c>0$ the sample median $\samplemedian(\bmx)$ satisfies  
\begin{align}
\lim_{D\to \infty} \Pr [ |\samplemedian(\bmx)-\median_X|\geq c ] = 0.
\end{align}
\end{lemma}

In words, \fref{lem:convergenceofmedian} implies that in the large-dimension limit, i.e., when $D\to\infty$, the sample median~$\samplemedian(\bmx)$ converges to the median $\median_X$.
Hence, by observing a large number of samples, which is possible in modern multi-antenna mmWave or OFDM systems, we can accurately estimate the true median. 

\subsection{Statistical Model for Complex-Valued Sparse Vectors}
\label{sec:statisticalmodel}
In order to derive and analyze the blind estimators proposed in \fref{sec:nonparametricestimators}, we need a suitable statistical model for the sparse signal $\bms$.
The statistical model should (i) have as few parameters as possible while being able to model a large class of complex-valued sparse vectors typically arising in communication systems and (ii) facilitate a theoretical analysis.
In what follows, we consider Bernoulli complex Gaussian (BCG) random vectors \cite{vila11a,rangan11a}, which allow control over the signal sparsity and the signal power.

\begin{definition}[BCG Random Vector] \label{def:BCG}
Each entry  in the sparse vector $\bms\in\complexset^D$ is nonzero with activity rate $p\in(0,1]$, and the nonzero entries are i.i.d.\ circularly-symmetric complex Gaussian with variance $\Eh$.
The PDF of each entry $s_d$, $d=1,\ldots,D$, is therefore given by
\begin{align}
\conditionaltextstyle f_S(s_d) = (1-p) \delta(s_d) + p \frac{1}{\pi \Eh} e^{-\frac{|s_d|^2}{\Eh}},
\end{align}
where $\delta(\cdot)$ is the Dirac delta function. 
\end{definition}

With this statistical model, the expected number of nonzero entries (the sparsity) is $K = pD$ and the average power of the sparse signal  vector $\bms$ corresponds to $\Es = \frac{1}{D}\Ex{}{\|\bms\|^2}$.

In \fref{sys:systemmodel1}, we assumed that the noise vector $\bmn$ is i.i.d.\ circularly-symmetric complex Gaussian with variance~$\No$ per complex entry. Hence, the PDF of each entry is given by $f_N(n_d) = \frac{1}{\pi\No}e^{-|n_d|^2/\No}$.
Consequently, the PDF of the noisy observation vector~$\bmy=\bms+\bmn$ in \fref{eq:inputoutputrelation} is as follows.

\begin{definition}[Noisy BCG Random Vector] \label{def:noisyBCG}
The PDF of the entries $y_d$, $d=1,\ldots,D$, of a BCG random vector per \fref{def:BCG} observed as  in~\fref{sys:systemmodel1}  is given by
\begin{align}
\conditionaltextstyle f_Y(y_d) \!=\! (1\!-\!p) \frac{1}{\pi \No} e^{-\frac{|y_d|^2}{\No}}   \conditionaltextstyle \!+ p \frac{1}{\pi (\No+\Eh)} e^{-\frac{|y_d|^2}{\No+\Eh}}\!.
\end{align}
\end{definition}

For this signal and observation model, we are now able to derive and analyze the proposed blind estimators.

\subsection{Analysis of \fref{est:noisevariance}}
\label{sec:noise_estimator_analysis}

We start with the blind noise variance estimator defined in \fref{est:noisevariance}.
We have the following key result. The proof is given in \fref{app:mainresult}.

\begin{theorem} \label{thm:mainresult}
Let $\bmy$ be a noisy BCG random vector with PDF as in \fref{def:noisyBCG} and with activity rate
\begin{align} 
p \leq\conditionaltextstyle  \frac{1/2 - e^{-2}}{ 1 - e^{-2} }  \approx 0.421. \label{eq:probabilitycondition}
\end{align}
Then, the average noise variance $\No$ satisfies 
\begin{align}  \label{eq:yummysandwichbound}
&\conditionaltextstyle  \frac{\median_Z}{ \min\left\{ \log\left(\frac{2-2p}{1-2p}\right), \log(2)(1+\SNR) \right\} }  \leq \No \notag \\
&\conditionaltextstyle  \qquad \qquad \qquad \qquad  \leq \frac{\median_Z}{\log(2)} \! \left( \!(1-p)+\frac{p^2}{p+\SNR}  \right)\!,
\end{align}
where $\median_Z$ is the median of an entry $z_d$ of $\bmz=\abssquared{\bmy}$.
\end{theorem}

It is important to realize that for $\SNR\to0$ or $p\to0$, the lower and upper bounds in \fref{eq:yummysandwichbound} coincide and equal $\median_Z/\log(2)$.
Furthermore, we reiterate that in the large-dimension limit, i.e., when \mbox{$D\to\infty$}, the sample median~$\samplemedian(\bmz)$ converges to the true median~$\median_Z$ as established by \fref{lem:convergenceofmedian}.
Thus, in these two cases the proposed estimate corresponds to the true noise variance $\No=\median_Z/\log(2)$.
Hence, by assuming the large-dimension limit, \fref{thm:mainresult} has the following key implications on  \fref{est:noisevariance}: 
(i) Since
\begin{align} \label{eq:additionalbound}
\conditionaltextstyle \No \overset{(a)}{\leq}  \frac{\median_Z}{\log(2)} \! \left( \!(1-p)+\frac{p^2}{p+\SNR}  \right) \overset{(b)}{\leq} \frac{\median_Z}{\log(2)} = \estimatedNo,
\end{align}
where $(a)$ follows from \fref{eq:yummysandwichbound} and $(b)$ holds for every valid value of $p$ and $\SNR$,
the proposed blind estimate $\estimatedNo$ bounds the average noise variance~$\No$ from above, i.e., we have developed a pessimistic estimator. 
(ii) By letting $p\to0$, we have that $\No=\median_Z/\log(2)$ and the inequalities in \fref{eq:additionalbound} hold with equality; this implies that the average noise variance will approach the value of the blind estimator for sparse signals (irrespective of the SNR).
(iii) By letting $\SNR\to0$, we have that $\No=\median_Z/\log(2)$ and the inequalities in~\fref{eq:additionalbound} hold with equality; this implies that the noise variance will approach the value of the blind estimator at low SNR (irrespective of the signal's sparsity).
In summary, the proposed noise variance estimate is pessimistic but accurate for sparse vectors or in low-SNR scenarios for high-dimensional data. 
While the above observations are only true for our blind estimator in the large-dimension limit and for the noisy BCG model in \fref{def:noisyBCG}, we will use simulations in \fref{sec:results}  to demonstrate the accuracy of \fref{est:noisevariance} for finite (and small) dimensions~$D$ and showcase an application example for channel-vector denoising in multi-antenna mmWave communication systems.

\subsection{Analysis of \fref{est:signalpower}}
\label{sec:signal_estimator_analysis}
For the blind estimate $\estimatedEs$ of the average signal power~$\Es$, we use the following standard result, which follows  from the fact that the entries of the vector $\bmz= \abssquared{\bmy}$ are i.i.d. with expected value of $\Ex{}{z_d}=\Es+\No$, $d=1,\ldots,D$. 

\begin{lemma} 
Let $\bmy$ be a noisy BCG random vector with PDF as in  \fref{def:noisyBCG}. Then, we have that
\begin{align}
\lim_{D\to\infty}\conditionaltextstyle \frac{1}{D} \|\bmy\|^2  \overset{a.s.}{\longrightarrow} \Es + \No,
\end{align}
where  $\overset{a.s.}{\longrightarrow}$ implies almost sure convergence. 
\end{lemma}
By defining the RV $W_D = \frac{1}{D}\|\bmy\|^2-\No$, we have that $\lim_{D\to\infty}W_D   \overset{a.s.}{\longrightarrow}  \Es$. By replacing the average noise power $\No$ by the blind estimate~$\estimatedNo$ in~\fref{eq:noiseestimator} from \fref{est:noisevariance}  and by clipping the result, we obtain \fref{est:signalpower} in~\fref{eq:signalpowerestimator}.
Since, in the large-dimension limit, the noise power estimate $\estimatedNo$ is overestimating the true average noise power,  the blind estimate $\estimatedEs$ in \fref{eq:signalpowerestimator}  tends to underestimate the signal power. From  \fref{thm:mainresult} it follows that for $p\to0$ or $\SNR\to0$, the blind signal power estimate is exact. 
While the above observations only hold  for $D\to\infty$, we showcase their accuracy for finite dimensions~$D$ in \fref{sec:results}.

\subsection{Analysis of \fref{est:SNR}}
\label{sec:SNR_estimator_analysis}
The blind SNR estimator is obtained by simply taking the ratio of $\estimatedEs$ in \fref{eq:signalpowerestimator} and $\estimatedNo$ in~\fref{eq:noiseestimator}.
For \mbox{$D\to\infty$}, the blind signal estimate {underestimates} the average signal power and the noise power estimate {overestimates} the average noise power, which means that the blind SNR estimate in \fref{eq:SNRestimateomg} underestimates the SNR. 
From  \fref{thm:mainresult} it follows that for $D\to\infty$ with either $p\to0$ or $\SNR\to0$ the blind SNR estimate is exact.
We provide simulation results for finite dimensions $D$ in \fref{sec:results}.

\subsection{Analysis of \fref{est:MSE}}
\label{sec:MSE_estimator_analysis}
In order to analyze \fref{est:MSE}, we first assume that the average noise power $\No$ is  known. For this  scenario, we can borrow the following two theorems from \cite{mirfarshbafan2019beamspace}.

\begin{theorem}[Thm.~1 of \cite{mirfarshbafan2019beamspace}]
\label{thm:MSEappox}
Let $\bmy$ be a noisy random vector of observations of $\bms$ as in \fref{sys:systemmodel1},
and apply a weakly differentiable function $\festnoargs:\complexset\to\complexset$ to the entries of $\bmy$ as in \fref{sys:systemmodel2}.
Then, Stein's unbiased risk estimate given by
\begin{align} \label{eq:complexSURE}
  \textit{SURE}=\, &\conditionaltextstyle \frac{1}{D}\|\fest{\bmy}-\bmy\|_2^2 - \No\nonumber\\ 
& \conditionaltextstyle +\frac{\No}{D} \sum_{d=1}^D\left(\frac{\partial\Re\{\fest{y_d}\}}{\partial\Re\{y_d\}}+\frac{\partial\Im\{\fest{y_d}\}}{\partial\Im\{y_d\}}\right)\!,
\end{align}
is an unbiased estimate of the MSE so that $\Ex{}{\textit{SURE}} = \MSE.$
\end{theorem}
Under the same assumptions, we have the following result which characterizes the behavior in the  large-dimension limit. 

\begin{theorem}[Thm.~3 of \cite{mirfarshbafan2019beamspace}] \label{thm:SUREconvergence}
In the large-dimension limit, i.e., when $D\to\infty$, $\textit{SURE}$ in \fref{eq:complexSURE} converges to the MSE in~\fref{eq:MSEexpression}, i.e., we have $\lim_{D\to\infty}\textit{SURE} = \MSE.$
\end{theorem}

These results imply that if $\No$ were known perfectly, one can estimate the MSE without knowing the sparse signal vector~$\bms$ in the large-dimension limit when $D\to\infty$. 
To obtain the blind version in \fref{est:MSE}, we have replaced the true average noise power $\No$ by its estimate $\estimatedNo$. 
Consequently, for $D\to\infty$ and either $p\to0$ or $\SNR\to0$, \fref{thm:mainresult} implies that $\estimatedNo$ will be exact, which implies that \fref{est:MSE} will be exact in this scenario. 
To demonstrate the efficacy of this estimator in  finite dimensions, we will show synthetic results and an application for mmWave channel-vector denoising in \fref{sec:results}.

% !TEX root = 20ICC_nonparametric.tex
% DO NOT REMOVE THE ABOVE COMMENT!

\newcommand{\trials}{10000} 
\newcommand{\figW}{0.715\columnwidth}

\begin{figure*}
	\centering
	\subfigure[Average noise power]{\includegraphics[width=\figW]{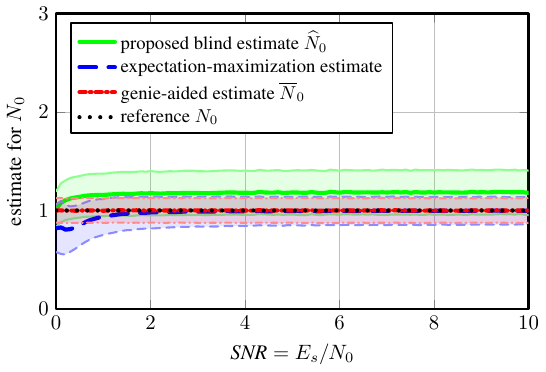}}
	\hspace{1.5cm}
	\subfigure[Average signal power]{\includegraphics[width=\figW]{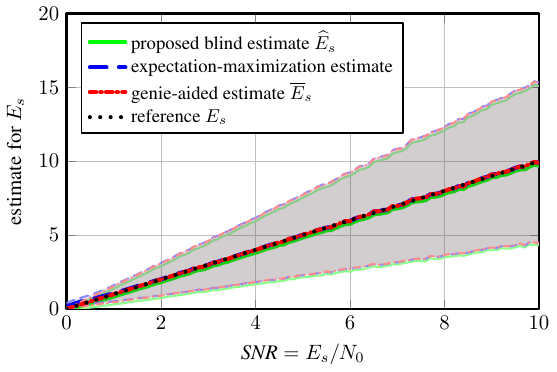}} \\
	\vspace{-0.1cm}
	\subfigure[Signal-to-noise ratio]{\includegraphics[width=\figW]{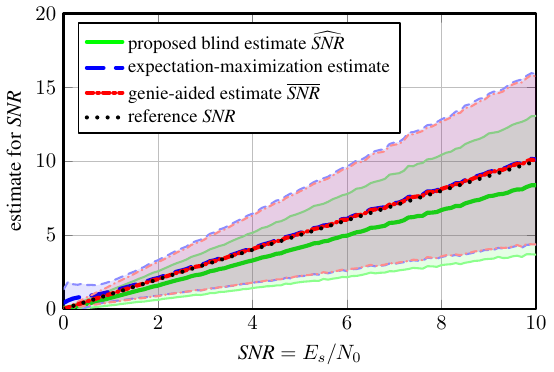}}
	\hspace{1.5cm}
	\subfigure[Mean-square error]{\includegraphics[width=\figW]{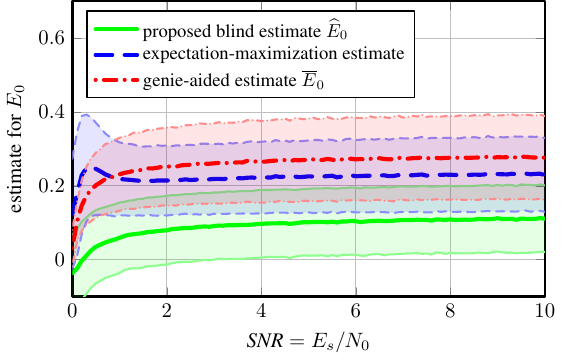}}
	\vspace{-0.1cm}
	\caption{Effect of varying the SNR on the proposed low-complexity blind estimators for the average noise power, signal power, SNR, and MSE.}
	\label{fig:varyingSNR}
\vspace{-0.35cm}	
\end{figure*}

\section{Numerical Results}
\label{sec:results}
 
\subsection{Synthetic Results}
We perform Monte--Carlo simulations with 10,000 trials to characterize the accuracy of the estimators proposed in \fref{sec:estimatorsusbection}.
We use the sparse signal model in \fref{def:noisyBCG} and fix $\No=1$, without loss of generality. We show the effect of the $\textit{SNR}=\Es/\No$ on the proposed estimators for an activity rate of $p=0.1$ and a dimension of $D=64$.  

\fref{fig:varyingSNR} shows the accuracy when varying the SNR. The results for our blind estimators are shown in green, where thick solid lines refer to the average performance and light green areas indicate the standard deviation.
For the MSE $\Eo$, we pick the soft-thresholding function $\fest{\bmy}=\eta(\bmy;\tau)$.
Furthermore, for each received vector $\bmy$, we adaptively select the denoising threshold $\tau\geq0$ that minimizes the estimated MSE $\estimatedEo$ as done in~\cite{ghods19a}.
As a comparison, we include an EM implementation specialized for circularly-symmetric complex Gaussian mixtures (average performance shown with a blue dashed line and standard deviation with a light blue area) with a maximum of 30 iterations and early stopping if the total parameter change is below $0.1$\%.
We also show the accuracy of genie-aided estimators (average performance shown with a red dash-dotted line and standard deviation with a light red area) that have separate knowledge of $\bmn$ and~$\bms$. 
We compute $\sampleEs = \frac{1}{D}\|\bms\|^2$, $\sampleNo = \frac{1}{D}\|\bmn\|^2$, $\sampleSNR= \frac{\sampleEs}{\sampleNo}$, $\sampleEo = \frac{1}{D} \|\eta(\bmy;\tau)-\bms\|^2$.  
The reference parameters used in our simulations are shown with black dotted lines. Note that there is no reference value for~$\Eo$, as the adaptive threshold prevents us from computing~$\Eo$ analytically.

From \fref{fig:varyingSNR}, we observe the following facts:
(i) The standard deviation of the proposed blind estimators is comparable to that of the genie-aided methods that have separate access to~$\bmn$ and $\bms$. 
(ii) Even though the sample size is small ($D=64$), our estimators are quite accurate with a standard deviation comparable to that of the genie-aided estimators; increasing~$D$ would further reduce the standard deviation of all considered estimators.
(iii) As predicted by our theory, the average noise power {is} overestimated {while} the signal power {and SNR are} underestimated. At low SNR, the three become exact.
(iv) For the considered scenario, the {blind MSE estimate has a negative offset.}
%is appreciably underestimated. 
However, the key requirement for adaptive parameter tuning is that the {estimated} MSE $\estimatedEo$ (which is the function to be minimized) has a similar shape as~$\Eo$ 
{and thus, for this purpose, we do not clip negative values.} 
As we show next, the MSE estimate still performs well in practice.

In comparison with EM, our method provides a less-accurate estimate at higher SNRs, but requires significantly lower complexity. 
 The complexity (in terms of the number of real-valued additions, real-valued multiplications, and exponentials) of EM is more than  $N(16D+12)+3D$ operations, where~$N$ is the number of EM iterations---the average number of iterations observed in our simulations ranges from $8$ to $28$ depending on the SNR. 
In contrast, our proposed median-based noise estimator has an average complexity of no more than $7.7D+9$ operations, when computing the median using quickselect~\cite{tibshirani09a}. 
Hence, our proposed blind estimator is more than $17\times$ less complex than EM (and avoids the evaluation of complex operations such as  exponentials and divisions), which renders our method suitable for (i) low-complexity parameter estimation and (ii) as a potential initializer for EM-based methods---the latter aspect is part of ongoing research.

\subsection{Application to Nonparametric Channel-Vector Denoising}
\begin{figure}
\vspace{-0.15cm}
	\centering
	\subfigure[Uncoded BER vs.~SNR]{\includegraphics[width=\figW]{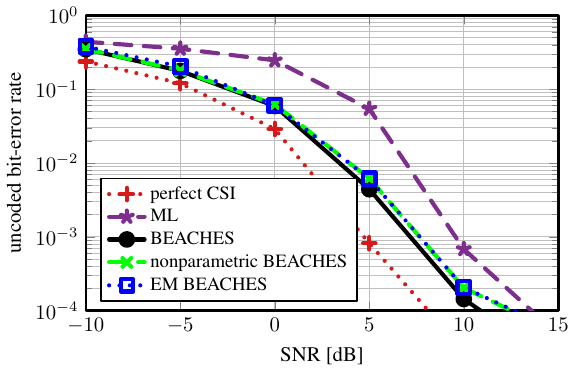}}
	\subfigure[MSE vs.~SNR]{\includegraphics[width=\figW]{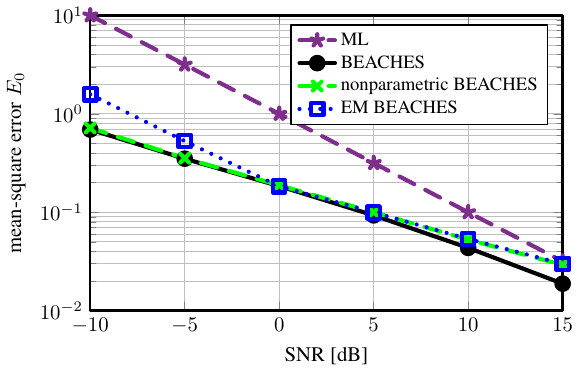}}
	\vspace{-0.1cm}
	\caption{Uncoded BER (a) and MSE (b) of mmWave channel estimation methods, including the nonparametric BEACHES variant which estimates the noise variance and denoising parameter directly from the receive vector.}
	\label{fig:mad}
\vspace{-0.35cm}
\end{figure}

We now show  an application of \fref{est:MSE} for beamspace channel estimation. 
As in \cite{ghods19a}, we simulate an uplink massive MIMO system in which $\U$ single-antenna user equipments (UEs) transmit pilots and data to a BS equipped with a uniform linear array of $\B$ antenna elements.
We consider  channel estimation with orthogonal pilots. The maximum likelihood (ML) estimate of the channel matrix is obtained by right-multiplying the received pilot sequence with the inverse of the orthogonal pilot matrix. In this case, the beamspace representation of the estimated channel matrix $\widetilde\bH$ is  given by $\widetilde\bH = \bH + \widetilde\bN$, where $\bH\in\complexset^{128\times8}$ is the beamspace channel matrix, $\widetilde\bN$ is complex Gaussian noise, and $\widetilde\bH$ is the beamspace ML channel estimate, which is a noisy observation of $\bH$.

Since electromagnetic waves at high frequencies experience strong attenuation and little scattering or diffraction, typical mmWave channels consist only of a small number of dominant propagation paths.
As each row index in the beamspace channel matrix corresponds to an angle-of-arrival to the BS, each column of the beamspace channel matrix $\bH$ (which is the \emph{channel vector} of one UE) will be sparse.
As a consequence, by writing each column of $\widetilde\bH$ as an independent equation, we can express the channel estimation problem in the form of \fref{sys:systemmodel1}, where each column of the channel matrix (that contains only few nonzero entries) corresponds to the sparse signal $\bms$. 
This observation implies that we can perform denoising in order to improve the ML channel estimate.

\fref{fig:mad} shows simulation results for 10,000 Monte--Carlo trials with line-of-sight (LoS) mmMAGIC QuaDRiGa mmWave channels~\cite{jaeckel2019quadriga}.
For different channel estimation methods we compute the MSE, and the bit-error rate (BER) with linear minimum MSE equalization using the estimated channels and uncoded 16-QAM transmission.
We simulate beamspace channel estimation (BEACHES) as in \cite{ghods19a}, which denoises the 
columns of the beamspace ML estimate $\widetilde\bH$ by applying the soft-thresholding function $\eta(x;\tau)$.
The thresholding parameter $\tau$ is adaptively selected for each noisy observation by minimizing SURE with perfect knowledge of the average noise power~$\No$ using an $\setO(D\log(D))$ algorithm.
We compare this to a new, nonparametric BEACHES variant, where we apply soft-thresholding denoising to the beamspace channel vectors and use the (nonparametric) threshold $\tau$ that minimizes \fref{est:MSE}, which is a nonparametric version of SURE.
We also include a variant that we call EM BEACHES, where we use \fref{est:MSE} but replace $\estimatedNo$ by the EM noise power estimate.
The methods described above, after denoising the beamspace channel vectors, use the inverse Fourier transform to obtain an antenna-domain channel estimate.
As a reference, we show the performance of perfect channel state information (CSI) that uses the ground truth (noiseless) channel vector, and maximum likelihood (ML) estimation that simply takes the inverse Fourier transform of the beamspace noisy observation~$\widetilde\bH$ as the antenna-domain channel estimate.

From \fref{fig:mad}, we observe that the nonparametric BEACHES algorithm achieves virtually the same performance as that of the original BEACHES algorithm which requires knowledge of~$\No$ (except at high SNR where \fref{est:noisevariance} tends to overestimate $\No$). We reiterate that the nonparametric BEACHES algorithm requires no parameters and exhibits exactly the same complexity as the original algorithm as the latter already sorts the entries of $\abssquared{\bmy}$, which we can reuse to compute \fref{est:noisevariance}.
EM BEACHES achieves higher (worse) MSE, as realistic channels deviate from the BCG model in \fref{def:BCG}, and exhibits higher complexity than our nonparametric method.

% !TEX root = 20ICC_nonparametric.tex
% DO NOT REMOVE THE ABOVE COMMENT!

 \section{Conclusions}
We have proposed blind estimators  for the average noise power, signal power, SNR, and MSE. Our estimators can be calculated in $\setO(D)$ time and only require the noisy observation vector, which avoids the need for additional pilot signals. 
We have analyzed our estimators for a complex Bernoulli-Gaussian sparsity model and evaluated their accuracy via simulations. 
Using a  channel-vector denoising task in multi-antenna mmWave systems, we have demonstrated that our blind estimators lead to a novel nonparametric denoiser that achieves comparable performance and the same complexity as BEACHES in \cite{ghods19a,mirfarshbafan2019beamspace} which requires knowledge of the average noise power. 
We believe that the proposed blind estimators find potential use in a large number of other applications  in wireless communication systems that contain sparse signals and require low complexity.

\appendices
% !TEX root = 20ICC_nonparametric.tex
% DO NOT REMOVE THE ABOVE COMMENT!

\section{Proof of \fref{thm:mainresult}}
\label{app:mainresult}
 
\subsection{Prerequisites}
In what follows, we need the distribution of $\bmz=\abssquared{\bmy}$.
Since the absolute-square entry of a circularly-symmetric complex Gaussian RV $Q$ with variance $E_q$ is exponentially distributed with CDF $F_Q(q)=1 - e^{-\frac{q}{{E_q}}}$, $q\geq0$, the CDF of each entry of the absolute-square noisy observation is as follows.

\begin{definition}[Noisy BCG Power RV] \label{def:noisyBCGsquared}
Let $\bmy$ be as in \fref{def:noisyBCG} and let $\bmz=\abssquared{\bmy}$.
Then, for $z\geq0$, the CDF of each  entry of $\bmz$ is given by
\begin{align} \label{eq:noisyBCGpower}
F_Z(z_d) = \, & (1-p)(1 - e^{-\frac{z_d}{\No}})  
 + p(1 - e^{-\frac{z_d}{\No+\Eh}}).  \
\end{align}
\end{definition}

\subsection{Upper Bounds on the Median}
We start with the following upper bound on the median $\median_Z$ of a noisy BCG power RV~$Z$ with CDF given in \fref{eq:noisyBCGpower}. 

\begin{lemma}  \label{lem:smollemma1}
For a noisy BCG power RV in \fref{def:noisyBCGsquared} with $p < 0.5$, the median is bounded  from above by
\begin{align}
\conditionaltextstyle \No\log\!\left(\frac{2-2p}{1-2p}\right)\! & \geq  \median_Z.
\end{align}
\end{lemma}

\begin{proof}
Using \fref{eq:noisyBCGpower}, we obtain the expression for the median of a RV $Z$ with CDF in \fref{eq:noisyBCGpower} according to equation \fref{eq:median}:
\begin{align}
(1-p)(1 - e^{-\frac{\median_Z}{\No}}) + p(1 - e^{-\frac{\median_Z}{\No+\Eh}}) & =\conditionaltextstyle  \frac{1}{2}.
\end{align}
Since the second term is nonnegative, we can omit it to obtain the following inequality:
\begin{align} \label{eq:expression0}
(1-p)(1 - e^{-\frac{\median_Z}{\No}})   & \leq\conditionaltextstyle \frac{1}{2}.
\end{align}
Note that this bound will be useful for vectors $\bms$ that are sparse, i.e., where $p$ is small. We can simplify \fref{eq:expression0} as follows
\begin{align}
\conditionaltextstyle \frac{\median_Z}{\No} & \leq \conditionaltextstyle -\log\left(\frac{1-2p}{2-2p}\right) \label{eq:logarithmcondition} 
\end{align}
which leads to an upper bound on the median $\median_Z$.
In order to take the logarithm in \fref{eq:logarithmcondition}, we require $p\in(0,1/2)$.
\end{proof}

\begin{lemma} \label{lem:smollemma2}
For a noisy BCG power RV $Z$ in \fref{def:noisyBCGsquared} with $p\leq \frac{1/2 - e^{-2} }{ 1 - e^{-2} }$, the median is bounded  from above by
\begin{align}
\log(2)(\No + \Es)  & \geq \median_Z.
\end{align}
\end{lemma}

\begin{proof}
From \fref{eq:noisyBCGpower}, we have that
\begin{align}
\conditionaltextstyle\frac{1}{2} = (1-p)e^{-\frac{\median_Z}{\No}} + p e^{-\frac{\median_Z}{\No+\Eh}}.     \label{eq:niceform}
\end{align}
Let us define the function $g(r)=e^{-1/r}$ with $r>0$. We can now rewrite \fref{eq:niceform} as follows:
\begin{align}
\conditionaltextstyle \frac{1}{2} = (1-p)g\!\left(\frac{\No}{\median_Z}\right)  + p g\!\left(\frac{\No+\Eh}{\median_Z}\right)\!. \label{eq:nicejensenexpression}
\end{align}
The function $g(r)$ is concave for $r \geq 1/2$, which holds when
\begin{align}
2\No  & \geq \median_Z. \label{eq:muleq2Ee} 
\end{align}
We first verify when \fref{eq:muleq2Ee} holds. From \fref{eq:niceform}, we have that
\begin{align}
\conditionaltextstyle\frac{1}{2} & \geq (1-p)e^{-\frac{2\No}{\No}} + p e^{-\frac{2\No}{\No+\Eh}} \\
\conditionaltextstyle\frac{1}{2} - e^{-2} & \geq   p \left(  1 - e^{-2} \right)\!,
\end{align}
which implies that the condition \fref{eq:probabilitycondition} ensures concavity of $g(r)$. 
By assuming that the condition \fref{eq:probabilitycondition} holds, we can use Jensen's inequality on the expression in \fref{eq:nicejensenexpression}  to get
\begin{align}
\conditionaltextstyle\frac{1}{2} &\leq  \conditionaltextstyle g\left((1-p)\frac{\No}{\median_Z} + p\frac{\No+\Eh}{\median_Z}\right) 
 =  e^{-\median_Z \frac{1}{\No + \Es } }.
\end{align}
We can now simplify this expression to
\begin{align}
\log(1/2)  & \conditionaltextstyle \leq -\median_Z \frac{1}{\No + \Es } \\
\log(2)(\No + \Es)  & \geq \median_Z, \label{eq:probabilityupperbound}
\end{align}
which is what we show in \fref{lem:smollemma2}. 
\end{proof} 

\subsection{Lower Bound on the Median}
We now establish the following lower bound on the median. 
\begin{lemma} \label{lem:smollemma3}
For a noisy BCG power RV $Z$ in \fref{def:noisyBCGsquared} with $p \in (0,1]$, the median is bounded  from below by
\begin{align} \label{eq:bound3}
\conditionaltextstyle \frac{\log(2)\No}{ (1-p)+\frac{p^2}{p+\SNR}  } &  \leq \median_Z.
\end{align}
\end{lemma}

\begin{proof}
Since the exponential CDF $F_Q(q)=1 - e^{-\frac{q}{{E_q}}}$ for ${E_q}\geq0$ is concave in $q$, Jensen's inequality leads to
\begin{align}
(1-p)(1 - e^{-\frac{\median_Z}{\No}}) + p(1 - e^{-\frac{\median_Z}{\No+\Eh}}) &\conditionaltextstyle  = \frac{1}{2} \\
1 - e^{-(1-p)\frac{\median_Z}{\No}-p\frac{\median_Z}{\No+\Eh}} & \conditionaltextstyle \geq \frac{1}{2}.
\end{align}
We can simplify this expression  to obtain the following bound
\begin{align}
\conditionaltextstyle \frac{1}{2} & \conditionaltextstyle\geq  e^{-(1-p)\frac{\median_Z}{\No}-p\frac{\median_Z}{\No+\Eh}} \\
\log\left(2\right) \No &\conditionaltextstyle \leq \median_Z \left( (1-p)+\frac{p^2}{p+\SNR} \right)\!,
\end{align}
which is what we have in \fref{eq:bound3}
\end{proof}

\subsection{Combining the Results}
Finally, we can combine  \fref{lem:smollemma1} with \fref{lem:smollemma2} and \fref{lem:smollemma3} to obtain the desired result in~\fref{eq:yummysandwichbound}.
%

%%%%%%%%%%%%%%%%%%%%%%%%%%%%%%%%%%%%%%%%%%%%%%%%%%%%%%%%%%%%%%%%%%%
%%%%%%%%%%%%%%%%%%%%%%%%%%%%%%%%%%%%%%%%%%%%%%%%%%%%%%%%%%%%%%%%%%%

\bibliographystyle{IEEEtran}
\balance
\bibliography{bib/confs-jrnls,bib/IEEEabrv,bib/publishers,bib/vipbib,bib/sm_ref}

\end{document}